\documentclass[11pt]{article}

\usepackage[latin1]{inputenc}
\usepackage[T1]{fontenc}
\usepackage[english]{babel}
\usepackage[dvips]{graphicx}
\usepackage{enumerate,amsthm,amsmath,amssymb,color,graphicx,bbm,float,framed}
\usepackage{fullpage}
\usepackage[square,sort,comma,numbers]{natbib}

\newcommand{\F}{\mathbb{F}}
\newcommand{\E}{\mathbb{E}}
\newcommand{\COMMENTZ}[1]{}

\newcommand{\Str}{\emph{Str}}
\newcommand{\Open}{\emph{Open}}
\newcommand{\Comm}{\emph{Comm}}
\newcommand{\CHSH}{\mathrm{CHSH}}

\newcommand{\Av}{\mathcal{A}_{c(v)}}
\newcommand{\Bv}{\mathcal{B}_{c(v)}}
\newcommand{\Alice}{\mathcal{A}}
\newcommand{\Bob}{\mathcal{B}}
\newcommand{\lv}{\mathrm{leaf}(v)}

\makeatletter
\newtheorem*{rep@theorem}{\rep@title}
\newcommand{\newreptheorem}[2]{%
\newenvironment{rep#1}[1]{%
 \def\rep@title{#2 \ref{##1}}%
 \begin{rep@theorem}}%
 {\end{rep@theorem}}}
\makeatother

\newtheorem{theo}{Theorem} 
\newtheorem{lemma}[theo]{Lemma}
\newtheorem{proposition}{Proposition}
\newtheorem{corollary}{	Corollary}

\newreptheorem{lemma}{Lemma}

\newtheorem{conj}[theo]{Conjecture}

\usepackage{epstopdf,hyperref}

\newcommand{\be}{\begin{equation}}
\newcommand{\ee}{\end{equation}}
\newcommand{\etal}{\textit{et al.} }
\newcommand{\eps}{\varepsilon}
\newtheorem{definition}{Definition}
\newtheorem{observation}{Observation}

\usepackage{mathtools}
\def\multiset#1#2{\ensuremath{\left(\kern-.3em\left(\genfrac{}{}{0pt}{}{#1}{#2}\right)\kern-.3em\right)}}

\usepackage[latin1]{inputenc}
\usepackage[T1]{fontenc}
\usepackage[english]{babel}
\usepackage[dvips]{graphicx}
\usepackage{enumerate,amsthm,amsmath,amssymb,color,graphicx,bbm,float,framed}
\usepackage{natbib}
\usepackage[noend]{algpseudocode}
\usepackage{amsmath}
\makeatletter
\def\BState{\State\hskip-\ALG@thistlm}
\makeatother

\newcommand{\FQ}{\mathbf{\mathsf{\mathbbm{F}_Q}}}

\newcommand{\zo}{\{0,1\}}

\title{Robust Relativistic Bit Commitment}

\date{}

\author{Kaushik Chakraborty \quad Andr\'e Chailloux \quad Anthony Leverrier \\ Inria Paris, France}

\begin{document}

\maketitle

\begin{abstract}
Relativistic cryptography exploits the fact that no information can travel faster than the speed of light in order to obtain security guarantees that cannot be achieved from the laws of quantum mechanics alone. Recently, Lunghi \textit{et al} [\textit{Phys.~Rev.~Lett.}~2015] presented a bit commitment scheme where each party uses two agents that exchange classical information in a synchronized fashion, and that is both hiding and binding. A caveat is that the commitment time is intrinsically limited by the spatial configuration of the players, and increasing this time requires the agents to exchange messages during the whole duration of the protocol. While such a solution remains computationally attractive, its practicality is severely limited in realistic settings since all communication must remain perfectly synchronized at all times. 

In this work, we introduce a robust protocol for relativistic bit commitment that tolerates failures of the classical communication network. This is done by adding a third agent to both parties. Our scheme provides a quadratic improvement in terms of expected sustain time compared to the original protocol, while retaining the same level of security.
\end{abstract}

\section{Introduction}

Bit commitment is a cryptographic primitive between two players Alice (the committer), and Bob (the receiver) who do not trust each other. A bit commitment 
protocol has two main phases: a \emph{commit phase} and an \emph{open} (or \emph{reveal}) \emph{phase}. Alice commits to a bit $d$ during the commit phase. We say that the protocol is \emph{hiding} if before the open phase, Bob has no information about $d$. During the open phase, Alice reveals $d$ to Bob, who wants to make sure that Alice didn't change her mind about the value of $d$, this is the \emph{binding} property.

It is well-known that bit commitment is impossible in the standard model \cite{BGK88}, even when allowing for quantum protocols \cite{May97,LC97}. In that case, it was shown that a protocol 
cannot be both hiding and binding. On the other hand, bit commitment becomes possible in the splitting agent model, where the two players Alice and Bob have a coalition of agents at their disposal: $\mathcal{A}_1, \ldots, \mathcal{A}_m$ for Alice, $\mathcal{B}_1, \ldots, \mathcal{B}_m$ for Bob.
The basic idea is to dispatch these agents in $m$ distant locations and restrict the information exchange between different locations.
This model has been extensively considered in the classical domain since the no communication assumption allows to implement many interesting cryptographic primitives: bit commitment \cite{BGK88}, oblivious transfer \cite{NP00} or protocols for private information retrieval \cite{GIKM98,KW04,GASA04}.

From a practical point of view, however, the no communication assumption is a bit difficult to justify. A convincing way to enforce it is to rely on the \emph{No Superluminal Signaling} (NSS) principle which states that no carrier of information can travel faster than the speed of light. In particular, an event in spacetime cannot be influenced by events which do not lie in its past causal cone.

The idea of using the NSS principle for cryptographic protocols originated in a pioneering work by Kent in 1999 \cite{Kent99} as a way to physically enforce the non communication constraint between the different agents of one party. The original goal of Kent was to bypass the no-go theorems for quantum bit-commitment \cite{May97,LC97}. Interestingly, this original protocol was classical and allowed for several rounds which increased the lifespan of the protocol. However, the protocol required to exchange messages whose length scaled exponentially in the number of rounds (\textit{i.e.}~the commitment time) and a feasible implementation was not possible for a large number of rounds.  A subsequent work \cite{Kent05} improved this scaling, but to our knowledge, no precise time/security tradeoff is available for this protocol.

More recently, quantum relativistic bit commitment protocols were developed where the parties exchange quantum systems, with the hope that combining the NSS principle with quantum theory will lead to more secure (but less practical) protocols \cite{Kent11,Kent12, KTH13}. In particular, the protocol \cite{Kent12} was implemented in Ref.~\cite{LKB13}.  We note that the scope of relativistic cryptography is not limited to bit commitment. For instance, there was recently some interest (sparked again by Kent) for position-verification protocols \cite{KMS10, LL11,Unr14} but contrary to the case of bit commitment, it was shown that secure position-verification is impossible both in the classical and the quantum settings \cite{CGM09,BCF14}.

The original idea of \cite{BGK88} was recently revisited by Cr\'epeau \etal \cite{CSST11} (see also \cite{sim07}). Based on this work, Lunghi \textit{et al.}~devised a multi-round bit commitment protocol involving only four agents, two for Alice and two for Bob \cite{LKB+15}. They managed to prove that this protocol, which we call the ``$\FQ$ protocol'' from now on, remains secure for several rounds, against classical attacks. Unfortunately, this proof was rather inefficient since the complexity of the protocol (the size of the messages the agents need to exchange at each round) scaled exponentially with the number of rounds. 
Recently, two papers improved the security proof and showed that the complexity of the protocol in fact only scales logarithmically with the number of rounds \cite{CCL15,FF15}, implying that the commitment time is essentially unlimited. This much better scaling shows that the protocol is quite practical, and a convincing experiment recently demonstrated the possibility of sustaining a commitment for 24 hours \cite{VMH+16}, consisting of $5 \times 10^9$ rounds. 
Although quite impressive, it should be noted that this implementation crucially used a 1 meter dedicated optical link between $\mathcal{A}_1$ and $\mathcal{B}_1$ (as well as between $\mathcal{A}_2$ and $\mathcal{B}_2$). 
In order to implement the protocol in a more realistic fashion, Alice and Bob's agents would need to communicate over a real telecom network, which is prone to rare failures, for instance delays in packet deliveries that would invalidate the no communication assumption and would cause the protocol to abort.

An important drawback of the $\FQ$ protocol is that it is not at all robust against losses, or delays. Indeed, for the bit commitment to succeed, it is crucial that the various agents communicate with perfect synchronization for all $k$ rounds of the protocol: if one agent fails to answer one challenge in time, then the whole protocol aborts. 
While this could be fine for small values of $k$, say $k \leq 10$, this is obviously disastrous for much larger values, for instance $k$ ranging in the millions or billions as in \cite{VMH+16}. For this reason, it is important to see whether some variant of the $\FQ$ protocol can be made tolerant against (a limited) amount of losses.
In this paper, we investigate one such variant where the original $\FQ$ protocol is modified so that both parties have now three agents at their disposal instead of two. 
We present the protocol in Section \ref{sec:protocols}. We prove its security against classical adversaries in Section \ref{sec:security} where we show that the security scales similarly as for the $\FQ$ protocol. 
Finally, in Section \ref{sec:analysis}, we show that the communication cost of the protocol is comparable to that of the $\FQ$ protocol but that its expected commitment time is quadratically improved.

\section{Description of the commitment schemes}
\label{sec:protocols}

A commitment scheme $\Pi = (COMM,OPEN)$ is the description of the protocol followed by the honest parties during both the commit and the open phases. All the protocols that we consider in this paper will be perfectly hiding and we will consequently only be interested in the binding property. Therefore, we only consider the case of a cheating Alice, 
which will be described through her cheating strategy  $\Str^* = (\Comm^*,\Open^*)$ in both phases of the protocol. 
The binding property we consider is the standard sum-property, that was also used 
in previous work regarding relativistic bit commitment \cite{LKB+15,FF15,CCL15}. 

\begin{definition}[Sum-binding]
\label{def1}
	We say that a bit commitment protocol $\Pi$ is \emph{$\eps$-sum-binding} if
	$$ \forall \ \Comm^*, \ \sum_{d = 0}^{1} \max_{\Open^*} \left( \Pr[\mbox{Alice successfully reveals } d \mid (\Comm^*,\Open^*)]  \right) \le 1 + \eps.$$
\end{definition}  

In this section, we describe successively the single-round protocol (with commitment time bounded by $\tau = D/c$ where $D$ is the distance between the distant locations and $c$ is the speed of light), the $\FQ$ multi-round protocol and finally our loss-tolerant protocol, the \emph{Tree protocol}.

For simplicity of analysis, we consider in this paper that all computations are performed instantaneously and that information travels at the speed of light. One could relax these assumptions by replacing $\tau$ by a smaller constant, but this would not change the various scalings of parameters and we therefore ignore this issue here. 

An important consequence of the fact that the protocols are perfectly hiding is that the spatial configuration of the agents needs only to be checked by Bob: in particular, it is sufficient for Bob to make sure that his agents are at a distance at least $D$ from each other. If this is the case, and if Alice's agents answer their challenges in time, then Bob can deduce that her agents are also separated by a distance $D$.

\subsection{The single-round protocol}

The single-round version of the protocol was introduced by Cr{\'e}peau \etal \cite{CSST11} (see also \cite{sim07}).
Both players, Alice and Bob, have agents $\mathcal{A}_1, \mathcal{A}_2$ and $\mathcal{B}_1, \mathcal{B}_2$ present at two spatial locations, $L_1$ and $L_2$, separated by a distance $D$. 
We consider the case where Alice makes the commitment. 
The protocol (followed by honest players) consists of four phases: preparation, commit, sustain and reveal. The sustain phase in the single-round protocol is trivial and simply 
consists in waiting for a time less than $\tau$, which is the time needed for light to travel between the two locations.

Overall the bit commitment protocol goes as follows. 
\begin{enumerate}
	\item \emph{Preparation phase}: $\mathcal{A}_1,\mathcal{A}_2$ (resp.~$\mathcal{B}_1,\mathcal{B}_2$) share a random number $a\in \FQ$ (resp.~$b \in \FQ$).
	\item \emph{Commit phase}: $\mathcal{B}_1$ sends $b$ to $\mathcal{A}_1$, who returns $y = a + d * b$ where $d \in \F_2$ is the committed bit. Here and everywhere in this paper, all operations are understood in $\F_Q$. 
	\item \emph{Sustain phase}: $\mathcal{A}_1$ and $\mathcal{A}_2$ wait for some time less than $\tau$.
	\item \emph{Reveal phase}: $\mathcal{A}_2$ reveals the values of $d$ and $a$ to $\mathcal{B}_2$ who checks that $y = a + d*b$.
\end{enumerate}

\subsection{The $\FQ$-protocol (multi-round, not loss-tolerant)}

The single-round protocol above was recently extended to a multi-round commitment scheme \cite{LKB+15}. The main idea to increase the commitment time is to delay the reveal phase and 
have $\mathcal{A}_2$ commit to the \emph{string} $a$ instead of revealing it. In fact, the new sustain phase will now consist of many rounds where the active agents (\textit{i.e.}~the 
agent of Alice who commits in that given round and the corresponding agent for Bob) alternate between locations $L_1$ and $L_2$.
Overall the $k$-round bit commitment protocol goes as follows (for $k$ even):
\begin{enumerate}
	\item \emph{Preparation phase}: $\mathcal{A}_1,\mathcal{A}_2$ (resp.~$\mathcal{B}_1,\mathcal{B}_2$) share $k$ random numbers $a_1,\dots,a_{k}$ (resp.~$b_1,\dots,b_{k}$) $\in \FQ$. 
	\item \emph{Commit phase} (round 1): $\mathcal{B}_1$ sends $b_1$ to $\mathcal{A}_1$, who returns $y_1 = a_1 + d * b_1$ where $d \in \F_2$ is the committed bit. 
	\item \emph{Sustain phase}: at round $j \leq k$, active Bob sends $b_j \in \FQ$ to active Alice, who returns $y_j= a_j+ b_j* a_{j-1} $.
	\item \emph{Reveal phase}: $\mathcal{A}_1$ reveals $d$ and $a_k$ to $\mathcal{B}_1$. $\mathcal{B}_1$ computes recursively $\alpha_0 = d$ and $\alpha_{i+1} = y_{i+1} - b_{i+1}*\alpha_i$ and checks that $\alpha_k = a_k$. If this is the case, Alice has successfully revealed the bit $d$.
\end{enumerate}

The main idea of the multi-round protocol is to delay the reveal phase in order to increase the commitment time. This delay is obtained by making the passive Alice commit 
to the value of the string she was supposed to reveal in the previous round. Since each round increases the total commitment time by a quantity equal to $\tau$ (modulo the time needed for the various algebraic manipulations in $\FQ$ that we ignore), one sees that the required number of rounds scales linearly with the commitment time one wishes to achieve.

We require that round $j$ finishes before any information about $b_{j-1}$ reaches the other Alice. For any $j$, this implies that Alice's active agent has no information about $b_{j-1}$. In particular, this means that $y_j$ is independent of $b_{j-1}$. This will be crucial in order to show security of the protocol.

\subsection{The Tree protocol (multi-round and loss-tolerant)}

In order to formulate a loss-tolerant variant of the $\FQ$-protocol, we require that each party has 3 agents located at three locations $L_1, L_2, L_3$ which are at least at a distance $D$ from each other. As in the $\FQ$ multi-round protocol, timing constraints are represented by rounds. In the original protocol, at each round, a pair of agents $(\Alice_i,\Bob_i)$ performs a communication round, consisting of a challenge $b_i$ from Bob's agent to Alice's agent and an answer $y_i$ from Alice's agent to Bob's. 

Our $k$-round Tree protocol is represented by the complete binary tree of depth $k$ with $2^{k+1}-1$ nodes (recalling that the tree with a single node has depth 0 by convention). The depth of a node $v$ is equal to the length $|v|$ of the string $v$. A node of the tree is a string $v$ of $j \leq k$ letters in the alphabet $\{\ell, r\}$, corresponding to left or right child. 
Let us denote by $V$ the set of all nodes of the tree, so that $|V| = 2^{k+1}-1$ and by $V^*$ the set of all internal nodes of the tree, that is nodes that are not leaves. Let us further denote $n_k = |V^*| = 2^{k}-1$ the cardinality of $V^*$.
The root of the tree is the empty string $\varnothing$. A given node $v$ of depth $j < k$ has two children, a left child $v\ell$ and a right child $vr$. A node $v$ of depth $j \geq 1$ has a unique parent $v(\mathrm{parent})$ and a unique brother $v(\mathrm{brother})$: indeed, if $v$ is of the form $wt$ with $t \in \{\ell, r\}$, then $v(\mathrm{parent})=w$ and $v(\mathrm{brother}) = w \bar{t}$ where $\bar{t}$ is the element of $\{\ell, r\}$ distinct from $t$.

To describe the Tree protocol, we need a 3-coloring $c$ of this complete binary tree of depth $k$. The coloring $c$ is a function
\begin{displaymath}
c: \left\{
\begin{tabular}{ccc}
$V$ & $\to$ & $\{1,2,3\}$\\
$v$ & $\mapsto$ & $c(v)$\\
\end{tabular}
\right.
\end{displaymath}
where $V$ is the set of all $2^{k+1}-1$ nodes in the tree, with the coloring property that for all $v$ of depth $j < k$, it holds that
\begin{align*}
\{ c(v), c(v\ell), c(vr) \} = \{1,2,3\}.
\end{align*}
The above constraints on the colors means that for any node $v$, the colors $c(v),c(v\ell)$ and $c(vr)$ are all different. In particular, two brothers have different color. This coloring will be used to assign a location $L_1, L_2$ or $L_3$ to each node of the tree. In other words, each node of the tree corresponds to a communication round taking place at the location $L_{c(v)}$ corresponding to the color $c(v)$ of the node $v$. 

  More precisely, each node $v$ of depth $j$ of the tree corresponds to a communication round with a challenge $b_v$ and an answer $y_v$ between agents $\Alice_{c(v)}$ and $\Bob_{c(v)}$ at round $j+1$. For a fixed depth, several nodes can have the same color $col$, the corresponding agents $\Alice_{col}$ and $\Bob_{col}$ will then perform all those communication rounds at this time $j+1$. The leaves of the protocol correspond to the revealing phase.

The new notion that appears in the context of loss-tolerant protocols is that of a \emph{dead or alive} node: we will say that a node $v$ fails (or is dead, or non responsive) if the corresponding agent $\mathcal{A}_{c(v)}$ fails to answer the challenge sent to her by $\mathcal{B}_{c(v)}$ within time $\tau$ at round $j = |v| - 1$. Alternatively, an agent is \emph{alive} (or responsive) if she succeeds in replying in time to the challenge. 
In order to account for this extra piece of information, we will denote by $\perp$ Alice's answer in case her agent is non responsive for a given node. Said otherwise, while Bob challenges will still be elements of $\FQ$, the answers of Alice's agents are elements of $\FQ \cup \{\perp\}$.

This failure can result from a global failure of the network for one agent $i$ for some rounds, in which case for all nodes $v$ of the corresponding depth with $c(v) = i$, we will have $b_v = \perp$. It may also happen that agent $\Alice_i$ may answer some queries in time but not some others, which will result in the corresponding nodes being alive or dead.
Of course, a cheating Alice will try to exploit such failures to increase to probability to successfully reveal the bit $d$ of her choice.

Overall the $k$-round Tree bit commitment protocol goes as follows (for $k \geq 2$):
\begin{enumerate}
	\item \emph{Preparation phase}: Agents $\mathcal{A}_i$ and $\mathcal{B}_i$ are located at $L_i$ for $i \in \{1,2,3\}$. Moreover, $\mathcal{A}_1,\mathcal{A}_2, \mathcal{A}_3$ (resp.~$\mathcal{B}_1,\mathcal{B}_2, \mathcal{B}_3$) share $n_k = 2^{k}-1$ random numbers $(a_v)_{v \in V^*} \in \mathbbm{F}_Q^{n_k}$ (resp.~$(b_v)_{v \in V^*}\in \mathbbm{F}_Q^{n_k}$). This means that the agents share random numbers for all the internal nodes of the tree (not for the leaves).
	Alice's agents also share $d \in \{0,1\}$ which is the committed bit.
	\item \emph{Commit phase} (round 1): $\mathcal{B}_{c(\varnothing)}$ sends $b_{\varnothing}$ to $\mathcal{A}_{c(\varnothing)}$, who returns $y_{\varnothing} = a_{\varnothing} + d* b_{\varnothing}$. If Bob's agent $\mathcal{B}_{c(\varnothing)}$ does not receive Alice's response before time $\tau$, then the protocol aborts.
	\item \emph{Sustain phase} (rounds 2 to $k$): at round $j+1 \leq k$, for each node $vt$ of depth $j+1$ (\textit{i.e.}~$|v|=j$ and $t \in \{\ell, r\}$), agent $\mathcal{B}_{c(vt)}$ sends $b_{vt} \in \FQ$ to $\mathcal{A}_{c(vt)}$ who returns $y_{vt} = a_{vt} + b_{vt}* a_{v}$. If $\mathcal{B}_{c(vt)}$ does not receive Alice's response within time $\tau$, the corresponding value of $y_{vt}$ is set to the value corresponding to a dead node, that is $y_{vt} = \perp$. When this is the case, the branch is considered to be dead, and Bob's agents stop sending challenges for that particular branch as soon as they know it is dead. 
	\item \emph{Reveal phase}: For each node $v = wt$ of depth $k$ (\textit{i.e.}~with $|w| = k-1$ and $t\in\{\ell, r\}$), Agent $\mathcal{A}_{c(v)}$ reveals $d$ and $a_{w}$ to $\mathcal{B}_{c(v)}$. Bob's agents check $(i)$ that for each depth $j < k$, the leftmost alive node of the tree has at least one child alive and if it's the case, then $(ii)$ that for the leftmost alive path $(v_0 = \varnothing , v_1, \ldots, v_k = v)$ in the tree, Bob's agents compute recursively the values 
$\alpha_{\varnothing} = y_{\varnothing} - b_{\varnothing}*d$, $\alpha_{v_i} =  y_{v_i} -b_{v_i} * \alpha_{v_{i-1}} $ and check that $\alpha_{v_k} =  a_{v_{k}}$. If both conditions are satisfied, then Alice has successfully revealed the bit $d$.
\end{enumerate}

\begin{figure}[h!]
\centering
 \includegraphics[scale = 0.6]{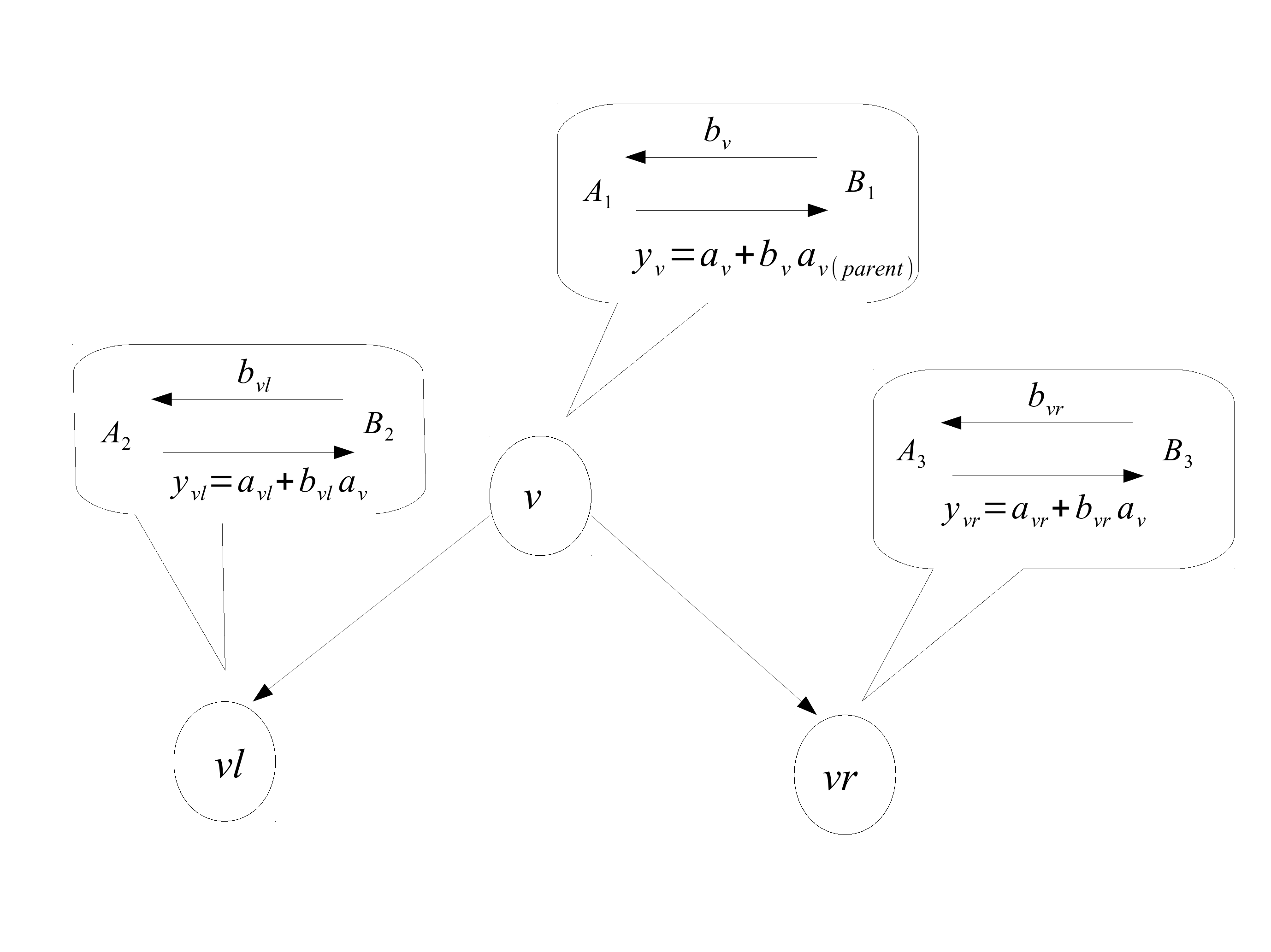}
\caption{Pictorial view for an internal node of the Tree protocol. Here the coloring is such that $c(v) = 1, c(v \ell)=2, c(v r)=3$.}
 \label{picture1}
\end{figure}

\textbf{Remark}:
Since only the values of the left-most alive branch matter for the verification step, it is useless in practice to keep other branches alive. A simple modification of the above protocol consists for Bob's agents to keep track of the leftmost alive branch and stop sending challenges for all other branches. We will analyze this in further detail in Section \ref{sec:analysis} where we investigate the communication cost of the Tree protocol.

\section{Security of the Tree protocol}
\label{sec:security}

The three protocols described above all share the property that they are perfectly hiding. Indeed, the role of the variables $a$'s shared by Alice's agents is to hide the value of $d$. If all the $a$'s are chosen uniformly at random in $\FQ$ which is the case if Alice follows honestly the protocol, then they provide a one-time pad of the secret and Bob's agents cannot obtain any information about the value of $d$ before the reveal phase. 

For this reason, our goal is to study whether these protocols are binding. In particular, this means that we will only be interested in the case where Bob is honest and follows the protocol, and Alice's agents might deviate from the protocol in order to reveal a bit that is not necessarily the one they had in mind during the commit phase. In this paper, we assume that Alice is classical, \textit{i.e.}, that her agents only share classical variables and not an entangled quantum state for instance. The question of proving security against a quantum adversary is left for future research.

Since Bob is assumed to be honest in the analysis, it means that his agents are correctly located at stations $L_1$, $L_2$ and $L_3$. In particular, there is no need for 
them to check where Alice's agents are located: it is sufficient to know that they responded in time to guarantee that for each round, each of them has to answer their 
own challenge without having access to the challenges sent to the other agents at the same round.

In all that follows, we consider without loss of generality a deterministic strategy for Alice for the $k$-round Tree protocol, in which any alive node has at least a live child.
Moreover, it is useful to understand what an optimal strategy for Alice looks like. Since only the leftmost alive branch matters in the reveal phase, at each round, Alice should make sure that the leftmost alive node has a live child, but she has some freedom to decide which one. It is easy to see that the best strategy is to always keep the right child reponsive and to decide whether the keep the left one alive or not based on the value of the challenge it receives. 
In other words, at each round, the left child of the leftmost alive child will decide either to answer its challenge (in which case, it will be the leftmost alive node at the next round), or to refuse to answer the challenge (in which case, its brother will become the leftmost alive node at the next round). 

\subsection{Sketch}
Our goal is to prove the security against a cheating Alice, on average over all of Bob's random strings $b$, which are drawn from the uniform distribution since Bob is honest. Depending on Alice's strategy and on those strings, the players will follow different leftmost paths in the tree.
The idea of the proof will be to use a recursive argument, similarly as in \cite{CCL15}. Informally, the proof will proceed as follows:

For each node $v$, we will keep track of a quantity $IP(v)$ (the Independence Parameter) that will quantify how independent $y_v$ is from $b_{v(\mathrm{parent})}$. For a fixed node $v$ of depth $j \leq k-2$, we will relate $IP(v)$ with $IP(v\ell)$ and $IP(vr)$. Then, if we define $IP_j$ to be the average independence parameter for nodes of depth $j$, we will use the previous relation to show that $IP_{j+1} \le IP_j + \frac{5}{4} \eps$ where $\eps = O(1/\sqrt{Q})$ is a security parameter.

Finally, in order to conclude, we will show that $IP_{k-1}$ corresponds exactly to Alice's cheating probability. Putting this together with the fact that $IP_0 \le \frac{1}{2} + \eps$, we will obtain the desired result.

In the above sketch, we omitted many discussions about the dependencies of the above quantities. In this section, we make the above argument formal, but defer several proofs to the Appendix. We will organize this section as follows.

In Subsection \ref{Section:Histories} below, we formally define several notions of history and of independence parameters that will be useful for our proofs. In Subsection \ref{Section:FinalCondition}, we relate the independence parameter $IP_{k-1}$ at the last round to the binding property of the protocol. Finally, in Section \ref{Section:TheProof}, we prove our recursive argument, and therefore prove the security of our protocol. The more technical details of the proof are deferred to the appendix.

\subsection{Notations \& Definitions}\label{Section:Histories}
For any $j \le k$, let $V_{\le j}$ be the set of nodes of depth at most $j$ and $V_{=j}$ the set of nodes of depth $j$. 

\begin{definition}
	For any integer $j \in [k]$, for any set $S \subseteq V_{\le j}$, let $H_j^S$ be the \emph{set of possible histories} of $S$, \textit{i.e.}~the set of possible commitment values $d \in \{0,1\}$ and strings
	$b_v \in \FQ$ for every $v \in S$. Since each $b_v$ is an element of $\F_Q$ for $v \in S$, we will identify an element of $H_j^S$ as an element of $\zo \times \F_Q^{ |S|}$.
\end{definition}
Let us note that in practice, Bob's agents stop sending challenges to nodes they know to be in a ``dead'' branch, which means that the corresponding $b_v$'s do not formally belong to $\F_Q$. For the security analysis, however, this is irrelevant since these nodes have no impact on the revealing phase of the bit commitment, which means that we can assume that these $b_v$'s are elements of $\F_Q$, so that the set of histories introduced above is well defined.

We also define $H_j:= H_j^{V_{\le j}}$ and $H_j^{-S}:= H_j^{(V_{\le j} - S)}$, which correspond respectively to the full history of nodes of depth at most $j$, 
and to the full history of such nodes, except for those in the set $S$. Moreover, we define $H_j^{S - Comm} := H_j^{S} \backslash \zo$ as the set $H_j^S$ where we remove the set of the committed bit. This is convenient when we need to talk about the history of the variables $b_v$'s only. In particular, we have $H_j = H_j^S \times H_j^{-S-Comm}$.
The set of all possible histories of the tree is $H_{k-1} := H_{k-1}^{V^{*}}$, since the leaf nodes only consist of Alice revealing (Bob's agents do not send any challenge for those nodes).

Since we assume without loss of generality that Alice follows a deterministic strategy, a history $h \in H_{k-1}$ induces Alice's answers $\{y_v\}_{v \in V^*}$ and therefore, if we run Alice's strategy on some history $h$, the state of all nodes, alive or dead, is fixed. Similarly, if we consider $h \in H_{j}$, this induces Alice's answers $\{y_v\}_{v \in V_{\le j}}$ and therefore, all nodes of depth at most $j$ are known to be either alive or dead.

\begin{definition}
	Let $v \in V_{\le j}$ and $h \in H_j$ be a node and a history. We say that $h$ is \emph{consistent with $v$} if when running Alice's strategy on $h$, the node $v$ is the left-most alive one at depth $\mathrm{depth}(v)$. We denote by $H_j(v) \subseteq H_j$ the set of histories consistent with $v$.
\end{definition}

Notice that we have 
$$ \bigcup_{v \in V_{=j}} H_j(v) = H_j \quad \textrm{and} \quad \forall v,v'\neq v \in V_{=j}, \ H_j(v) \cap H_j(v') = \varnothing, $$
which simply states that each history up to depth $j$ is consistent with exactly one node of $V_{=j}$.

\begin{definition}
	For $v \in V_{\le j}$, $S \subseteq V_{\le j}$ and $h_1 \in H_j^S$, we say that $h_1$ is \emph{consistent with $v$} if there exists $h_2 \in H_j^{-S-Comm}$ such that 
	$(h_1,h_2) \in H_j(v)$. We denote by $H_j^S(v) \subseteq H_j^S$ the set of $h_1 \in S$ consistent with $v$.
\end{definition}

By construction of the protocol, if Alice successfully reveals a value at the end, it means that for all rounds, the leftmost alive node has an alive child. In particular, this implies that the prefix of the leftmost alive branch doesn't change during the execution of the protocol: if $v$ be the leftmost alive node at depth $\mathrm{depth}(v)$ for a given $H^S_{depth(v)}(v)$, then it remains the leftmost alive node at depth $\mathrm{depth}(v)$ for any future history $H^S_j(v)$ 
with $j > \mathrm{depth}(v)$. We therefore have that for any non root node $v \in V_{\le j}$ and set $S \subseteq V_{\le j}$, $H_j^S(v) \subseteq H_j^S(w)$ where $w$ is the parent of $v$.

\begin{definition}
	For a fixed vertex $v \in V_{\le j}$,  a set $S \subseteq (V_{\le j} - \{v\})$ and a history $h \in H_j^S(v)$, let $B_j^h(v) := \{b_v \in F_Q: (h,b_v) \in H_j^{S \cup \{v\}}(v)\}$ be the set of values for $b_v$ for which node $v$ answers in time. Equivalently, $\F_Q - B_j^h(v)$ is the set of questions for which node $v$ will be non responsive, according to Alice's strategy and the history $h$.
\end{definition}
Note that if $v = wl$ is the left child of the leftmost alive node at depth $\mathrm{depth}(v)-1$, then $B_j^h(v)$ is the set of values in $\F_Q$ for which $v$ chooses to respond in time for 
Alice's strategy. On the other hand, if $b_v \not\in B_j^h(v)$, then the node chooses to be non responsive, and the leftmost alive node at that round becomes the right brother of $v$. Notice that $B_j^h(v)$ is independent of $b_w$.

\begin{definition}
\label{defz}
For $j \leq k$, we define the random variable $Z_j$ which takes value $v \in V_{=j}$ with probability $\frac{H_j(v)}{H_j}$. This random variable corresponds to the node that is the leftmost alive node at depth $j$.
\end{definition}

For each node $v$, let us recall that $\Av$ (resp.~$\Bv$) refers to Alice's (resp.~Bob's) agent at that node.

\begin{definition}
For any node $v \in V_{=j}$, let $Acc(v) \subseteq V_{\le j}$ be the set of nodes containing history information accessible to $\Av$, including the value of the commitment. 
\end{definition}
Crucially, the relativistic constraints impose that $v(parent),v(brother) \notin Acc(v)$.

Let us consider a vertex $v_j$ of depth $j$ and a history $h$ consistent with $v_j$. 
The leftmost alive path up to depth $j$ has the form $(v_0=\varnothing, v_1, \ldots, v_j)$. Recall that the variables $\alpha_{v_i}$ are recursively defined for $i \leq j$ by
\begin{equation}
\alpha_{v_i} := \left\{
\begin{array}{cl}
y_{v_0} - b_{v_0}*d & \text{if} \quad i=0,\\
 y_{v_{i}} - b_{v_{i}}* \alpha_{v_{i}(parent)} & \mathrm{otherwise}. \label{alpha}
\end{array}
\right.
\end{equation}

Recall also that $\alpha_{v_j}$ and $y_{v_j}$ are functions of the history $H_j$ since Alice's strategy is deterministic.

Similarly as in \cite{CCL15}, we introduce a quantity $IP$ which is the independence parameter between a variable and a function (or a family of functions). Intuitively, this quantity is large if the function is independent of the variable and close to 0 otherwise. In particular, it quantifies how well the function can be approximated by another function that does not depend on the given variable. This is relevant here since in a cheating strategy, Alice's agent tries to answer to Bob's challenge without knowing the value of the challenge sent to her parent, and she wins if she manages to give an answer that depends on that specific challenge. 
\begin{definition}
	For any integer $j \leq k-1$, any family of functions $\{g_v:  H_j^{Acc(v)}(v) \to \F_Q\}_{v \in V_{=j}}$, we define
	\begin{align}
	\label{def:IPj}
IP_j(\{g_v\}_{v \in V_{=j}}):= \E_{v \leftarrow Z_j} \E_{h \leftarrow H_j^{- \{v\}}(v)} \E_{b_v \leftarrow B_j^h(v)}[g_v(d,h) == \alpha_v(d,h,b_v)], 
\end{align}
where $g_{v}(d,h) == \alpha_{v}(d,h_,b_{v})$ represents the variable that equals $1$ if the equality $[g_{v}(d,h) = \alpha_{v}(d,h,b_{v})]$ holds and $0$ otherwise. Moreover, the notation $ \E_{v \leftarrow Z_j}$ corresponds to the expectation over the possible values $v$ of the random variable $Z_j$, and similarly for the other expectations. 
\end{definition}
Intuitively, this quantity is simply the expectation that Alice's agent (at round $j+1$) gives an answer consistent with the value $(\alpha_v$) expected by Bob's agent, for the leftmost alive node, when averaging over all possible histories: the restriction on Alice's strategy is that her agent at round $j+1$ does not know the value of $b_v$ at round $j$. 
Note here that in the above definition, the function $g$ takes as inputs elements more history elements than those in $H_j^{Acc(v)}(v)$. The function $g$ will 
simply disregard those inputs. We added them for notational simplicity but we will use later the fact that the outcome $g_v(d,h)$ actually depends only on the history elements of $H_j^{Acc(v)}(v)$.

We are finally in position to define the $IP$ parameter at depth $j$. 
\begin{definition} For $j \leq k-1$, the $IP$ parameter at depth $j$ is
\begin{align}
IP_j:= \max_{\{g_v\}_{v \in V_{=j}}} IP_j(\{g_v\}_{v \in V_{=j}}).
\end{align}
\end{definition}
In the next subsection, we provide some motivation for this definition by showing that $IP_{k-1}$ corresponds to Alice's cheating probability. This can be understood intuitively because $IP_{k-1}$ quantifies how well the agents of Alice at the $k^{\mathrm{th}}$ round (\textit{i.e.}~those you reveal the bit value) can give an answer consistent with Alice's agent's answer at the previous round.

\subsection{Final Condition}\label{Section:FinalCondition}
\begin{proposition}
	\label{prop:main}
	The $IP$ parameter satisfies the following bound:
	$$ 1 + \eps_k \le 2IP_{k-1}$$ where $\eps_k$ is the binding security parameter of the $k$-round protocol. 
\end{proposition}
\begin{proof}
Let $P^*_A$ be Alice's cheating probability. Let $P^*_{A | v}$ be Alice's cheating probability when the leftmost alive node at depth $k-1$ is $v$. We have by definition $P^*_A = \E_{v \leftarrow Z_{k-1}}[P^*_{A | v}]$. Let $\lv$ be the associated leaf that will be used for the reveal phase: $\lv = v \ell$ if $v \ell$ is alive, otherwise $\lv = v r$. Let $(a_{\lv}, d)$ be Alice's output for that leaf. Recall that Bob then checks whether $\alpha_v = a_{\lv}$ where $\alpha_v$ is computed recursively as in Eq.~\ref{alpha}. Bob's checking procedure implies that
\begin{align*}
P^*_{A | v} & = \E_{h \leftarrow H_j^{- \{v\}}(v)} \E_{b_v \leftarrow B_j^h(v)}[a_{\lv}(h) == \alpha_v(h,b_v)] \\
& \le \E_{h \leftarrow H_j^{- \{v\}}(v)}[\max_{g_v :H_j^{Acc(v)}(v) \to \F_Q} \{\E_{b_v \leftarrow B_j^h(v)}[g_v(h) == \alpha_v(h,b_v)]\}] =: IP_{k-1}(v)
\end{align*}
where we averaged over all histories giving $v$ as the leftmost node of depth $k-1$. From there, we have 
$$ P^*_A = \E_{v \leftarrow Z_{k-1}}[P^*_{A | v}] \leq \E_{v \leftarrow Z_{k-1}}[ IP_{k-1}(v)] = IP_{k-1}$$
By definition of the binding property, it holds that $P^*_A = \frac{1}{2} (1+ \eps_k)$, which yields the desired result.
\end{proof}
Proposition \ref{prop:main} shows that it is sufficient to prove a good upper bound on $IP_{k-1}$ in order to show that the bit-commitment protocol is binding.

\subsection{Bounding the value of $IP_{k-1}$}
\label{Section:TheProof}
Our goal is now to bound the value of $IP_{k-1}$. For this, we will use a recursive argument to bound $IP_j$ for all $j \le k-1$. Before that, we start by finding an expression for $IP_j$ that is suitable for a recursive analysis. 
Consider a node $v$ of depth $j \leq k-2$.
For a fixed history $h_0 \in H_{j+1}^{-\{v,v\ell,vr\}}(v)$,  two nodes $v$ and $vt$ (with $t \in \{\ell, r\}$), we define the quantity $IP^{h_0}_{vt}$:
\begin{equation}\label{Equation:IP}
IP_{vt}^{h_0}:= \max_{g  : \FQ \rightarrow \FQ}\E_{b_v \leftarrow B_j^{h_0}(v)} \E_{b_{vt} \leftarrow B_{j+1}^{h_0,b_v}(vt)} [g(b_v) == \alpha_{vt}(h_0,b_v,b_{vt})],
\end{equation}
where $vt$ is a child of node $v$.
We show the following: 
\begin{proposition}
\label{lem:lemma1}
 For all $j \leq k-2$, it holds that:
	$$
	IP_{j+1} = \E_{v \leftarrow Z_j} \E_{h_0 \leftarrow H_{j+1}^{-\{v,v\ell,vr\}}(v)}\E_{t \leftarrow T(v|h_0)} [IP_{vt}^{h_0}],
	$$
	where $T(v|h_0)$ is the function that outputs $t \in \{\ell, r\}$ if the leftmost alive child of $v$ is $vt$.
\end{proposition}
The proof of this proposition is based on elementary manipulations of the expected values and is presented in detail in Appendix \ref{Apppendix:ProofOfSumInversions}.

We can now proceed to bounding $IP_j$.
We first consider the base case where $j=0$.

\begin{lemma}
\label{lemma:ipvnot}
$$IP_{0} \le \frac{1}{2} + \sqrt{\frac{2}{Q}}.$$
\end{lemma}
\begin{proof}
	This property was already proven in \cite{CCL15}. For completeness, we reproduce this proof using the notations of the present paper in Appendix \ref{Appendix:j=0}.
\end{proof}

\begin{lemma}
\label{lemma:ipvth}
For every node $v \in V_{=j}$, $t \in \{\ell, r\}$ and history $h_0  \in H_{j+1}^{-\{v,v\ell,vr\}}(v)$ it holds that:
	$$IP_{vt}^{h_0} \le IP_{v}^{h_0} + \sqrt{\frac{2}{|B_{j+1}^{h_0}(vt)|}}.$$
where we slightly abuse notation by defining $IP_{v}^{h_0} := \max_{g} \E_{b_v \leftarrow B_j^{h_0}}[g = \alpha_v(h_0,b_v)]$.
\end{lemma}
The reason we say we slightly abuse notation is the discrepancy on what is fixed between this definition and the one in Equation \ref{Equation:IP}. Notice that we have 
$$IP_j = \E_{v \leftarrow Z_j} \E_{h_0 \leftarrow H_{j+1}^{-\{v,v\ell,vr\}}} [IP_{v}^{h_0}].$$
\begin{proof}
	We prove here Lemma \ref{lemma:ipvth}.
 As in \cite{CCL15}, we use the Alice's cheating strategy to come up with a strategy for a variant of the $\CHSH$ game with inputs and outputs in $\FQ$ instead of $\mathbbm{F}_2$. Then upper bounds on the classical value of this $\CHSH$ variant allow us to bound the value of $IP$.
 
 The class of $\mathrm{CHSH}_Q(p)$ games was introduced in \cite{CCL15} in order to analyze the security of the $\FQ$ protocols. These are simply two-party nonlocal games between Adeline and Bastian who respectively receive inputs $x, y\in \F_Q$ and output $a, b \in F_Q$. Here $x$ is drawn from the uniform distribution while $y$ is drawn according to a probability distribution $\{p_y\}_{y\in \F_Q}$ such that $\max_y p_y \leq p$.
Adeline and Bastian win the game if $a + b = x *y$ in $\FQ$. 
Let us define a slight variant of these games where the only difference is now that Adeline's inputs are drawn uniformly from a subset $S$ of $\F_Q$. We denote this class of games by  $\mathrm{CHSH}^S_Q(p)$.

We start with Equation \ref{Equation:IP}:

\begin{align*}
IP_{vt}^{h_0} = \max_{g  : \FQ \rightarrow \FQ}\E_{b_v \leftarrow B_j^{h_0}(v)} \E_{b_{vt} \leftarrow B_{j+1}^{h_0,b_v}(vt)} [g(b_v) == \alpha_{vt}(h_0,b_v,b_{vt})].
\end{align*}

 We  write $\alpha_{vt}(h,b_v,b_{vt}) = y_{vt}(h,b_{vt}) + b_{vt} * \alpha_{v}(h,b_v) $. From there, we can see that the dependence in $b_v$ of the function $\alpha_{vt}(h,b_v,b_{vt})$ lies only in the function $\alpha_{v}(h,b_v)$. Therefore, we can write 
 
 \begin{equation}
 IP_{vt}^{h_0} = \max_{g  : \FQ \rightarrow \FQ}\E_{b_v \leftarrow B_j^{h_0}(v)} \E_{b_{vt} \leftarrow B_{j+1}^{h_0,b_v}(vt)} [g(\alpha_v(h_0,b_v)) == \alpha_{vt}(h_0,b_v,b_{vt})].
 \end{equation}
 Let $\mathcal{G}^{h_0}$ be the function $g$ that maximizes the above expression. In order to end the proof, we perform the following steps: (1) we define an entangled game that will be an instance of some $\CHSH_Q^S$ game for some $S$, (2) we construct a cheating strategy for this game using the functions $y_{vt}$ and $G^{h^0}$ and finally (3) we use the known bounds on $\CHSH_Q^S$ to derive a bound on $IP^{h_0}_{vt}$. \\
 
 We consider the following game between two players Adeline and Bastian: 
 \begin{itemize}
 	\item Adeline receives a random element $X \in B_{j+1}^{h_0}(vt)$. Bastian receives an element $Y \in \FQ$ such that $\Pr[Y = c] = \Pr_{b_v}[\alpha_v(h,b_v) = c]$.
 	\item Their goal is to respectively output $A$ and $B$ in $\F_q$ such that $A + B = X * Y$
 \end{itemize}
Recall that $IP_v^{h_0} = \max_{c} \Pr_{b_v \leftarrow B^{h_0}_j(v)} [\alpha_v(h,b_v) = c]$. Since Adeline has no information about $b_v$, her probability of guessing $Y$ is upper bounded by $IP_v^{h_0}$. This means that the two player game we study is an instance of $\CHSH_Q^{B_{j+1}^{h_0}(vt)}(IP_v^{h_0})$. We know from Lemma \ref{lemma} (proven in Appendix \ref{games}) the following upper bound on the classical value of such a game: 
$$ \omega \Big(\CHSH_Q^{B_{j+1}^{h_0}(vt)}(IP_v^{h_0}) \Big) \le IP_v^{h_0} + \sqrt\frac{2}{|B_{j+1}^{h_0,b_v}(vt)|}.$$
We now use Alice's cheating strategy to derive a strategy for the above game. Adeline outputs $A = y_{vt}(h_0,X)$ and Bastian outputs $B = - \mathcal{G}^{h_0}(Y)$. We can lower bound the value of the game as follows:
\begin{align*}
\omega(\CHSH_Q^{B_{j+1}^{h_0}(vt)}(IP_v^{h_0})) & \ge \Pr_{X,Y} [A+B = X * Y] \\
& \ge \Pr_{X,Y} [y_{vt}(h_0,X) - \mathcal{G}^{h_0}(Y) = X * Y] \\
& = \Pr_{X,b_v} [y_{vt}(h_0,X)  - \mathcal{G}^{h_0}(\alpha_v(h_0,b_v)) = X * \alpha_v(h_0,b_v)] \\
&  = \Pr_{X,b_v} [\alpha_{vt}(h,b_v,X) + (\alpha_v(h_0,b_v) * X) -\mathcal{G}^{h_0}(\alpha_v(h_0,b_v))  = (X * \alpha_v(h_0,b_v))] \\
& = \Pr_{X,b_j} [\alpha_{vt}(h,b_v,X) = \mathcal{G}^{h_0}(\alpha_v(h_0,b_v))]\\
& = IP^{h_0}_{vt}.
\end{align*}
Combining the upper and the lower bound on $\omega(\CHSH_Q^{B_{j+1}^{h_0}(vt)}(IP_v^{h_0}))$, we conclude that 
$$ IP^{h_0}_{vt} \le IP_v^{h_0} + \sqrt\frac{2}{|B_{j+1}^{h_0}(vt)|}.$$
\end{proof}

We are now ready to prove the recurrence relation.
\begin{proposition} 
\label{prop:rec}
For $j \leq k-2$, it holds that:
	$$IP_{j+1} \le IP_j + \frac{5}{4}\sqrt{\frac{2}{Q}}.$$
\end{proposition}
\begin{proof}
For $v \in Z_{j},h_0 \in {H_{j+1}^{-\{v,v\ell,vr\}}}$, the probability that Alice is responsive at node $v\ell$, or equivalently, that $v\ell$ is the leftmost alive node at round $j+1$, is $\Pr[T(v|h_0) = \ell] = \frac{|B_{j+1}^{h_0}(v\ell)|}{Q} =: P_{h_0}$. 
Proposition \ref{lem:lemma1} gives:
\begin{align}
IP_{j+1} & = \E_{v \leftarrow Z_j} \E_{h_0 \leftarrow H_{j+1}^{-\{v,v\ell,vr\}}(v)}\E_{t \leftarrow T(v|h_0)} [IP_{vt}^{h_0}] \nonumber\\
& = \E_{v \leftarrow Z_j} \E_{h_0 \leftarrow H_{j+1}^{-\{v,v\ell,vr\}}(v)} [P_{h_0} IP_{v\ell}^{h_0} + (1 - P_{h_0}) IP_{vr}^{h_0}] \nonumber
\end{align}

We use Lemma \ref{lemma:ipvth} in order to bound $IP^{h_0}_{vl}$ and  $IP^{h_0}_{vr}$. We have by definition $|B^{h_0}_{j+1}(vl)| = P_{h_0}Q$ and $|B^{h_0}_{j+1}(vr)| = Q$. From there, we have 
\begin{align}
IP_{j+1} & = \E_{v \leftarrow Z_j} \E_{h_0 \leftarrow H_{j+1}^{-\{v,v\ell,vr\}}(v)} \left[P_{h_0}\left(IP_{v}^{h_0} + \sqrt{\frac{2}{P_{h_0}Q}}\right) + (1 - P_{h_0}) \left(IP_{v}^{h_0} + \sqrt{\frac{2}{Q}}\right) \right] \label{eq:prob}\\
& = \E_{v \leftarrow Z_j} \E_{h_0 \leftarrow H_{j+1}^{-\{v,v\ell,vr\}}(v)} \left[IP_{v}^{h_0} + (1 +\sqrt{P_{h_0}} - P_{h_0})\sqrt{\frac{2}{Q}}\right]\nonumber\\
& \le \E_{v \leftarrow Z_j} \E_{h_0 \leftarrow H_{j+1}^{-\{v,v\ell,vr\}}(v)} \left[IP_{v}^{h_0} + \frac{5}{4}  \sqrt{\frac{2}{Q}}\right] \label{eq:fct} \\
& = IP_j + \frac{5}{4}  \sqrt{\frac{2}{Q}}\nonumber
\end{align}
where we used the bound $(1 +\sqrt{P} - P) \leq \frac{5}{4}$ for $P \geq 0$ in Eq.~\ref{eq:fct}.
\end{proof}

Combining Propositions \ref{prop:main}, \ref{prop:rec} and Lemma \ref{lemma:ipvnot} gives our main result.
\begin{corollary} The $k$-round Tree protocol is $\eps_k$-sum-binding with
	$$\eps_k \le \frac{5k}{\sqrt{2Q}}.$$
\end{corollary}

This scaling is very close to the one of the $\FQ$ protocol for which the binding parameter is upper bounded by $2\sqrt{2} k/\sqrt{Q}$ according to Ref.~\cite{CCL15}. 

\section{Loss tolerance and communication cost of the Tree protocol}
\label{sec:analysis}

The main point of considering the Tree protocol instead of the simpler $\FQ$-protocol is that it displays some loss tolerance. 
In this section, we consider a very simple model of loss and evaluate the performance of the Tree protocol compared to the $\FQ$-protocol.

For this, we assume that in the honest case, each station (corresponding to a couple $\mathcal{A}_i, \mathcal{B}_i$) dies with some probability $p$ at each round of the protocol. This process is taken to be independent and identical. Moreover, we consider the scenario where a dead station remain dead for a time $m \tau$, where $m$ is some small integer such that $m \ll k$ and $mp \ll 1$. 
This loss model could of course be refined, for instance by adding correlations between the various probabilities of dying for modeling a global network failure for example, or by taking the dead time to be a random variable as well, but our simplified model allows for a more straightforward comparison of the different protocols and arguably already captures the behavior of realistic failures due to loss in bit commitment protocols. 

\begin{observation} \label{win-condition}
In the honest scenario where all players follow the protocol but losses are allowed, the Tree protocol protocol aborts if and only if two stations are dead at the same time (except at the first round).
\end{observation}

\begin{proposition}\label{prob-ok}
Provided that $mp \ll 1$ and $m \ll k$, the probabilities that the $k$-round $\FQ$ and Tree protocols don't abort are given by
\begin{align}
P_\mathrm{ok}(\FQ) &= (1-p)^k \\
P_\mathrm{ok}(\mathrm{Tree}) &= (1-q)^k  
\end{align}
with $q = 3 (mp)^2+(mp)^3$.
\end{proposition}

\begin{proof}
Let us first consider the $\FQ$ protocol: it aborts as soon as one station dies. At each round, a honest Alice responds in time with probability $1-p$. Since these events are assumed to be independent, the probability that Alice responds in time for the full protocol, that is, all $k$ rounds, is $P_\mathrm{ok}(\FQ) = (1-p)^k$.

In the Tree protocol, each station is non-responsive at a given round $i \geq m$ with probability $mp$ if we assume that $mp \ll 1$: this is the probability that the station died during any of the $m$ previous rounds. The probability that at least two stations are alive at a given round is equal to the probability that at most one of the three stations is non-responsive, that is $(mp)^3 + 3 (mp)^2 =q$.
It follows that the probability that the Tree protocol does not abort is $(1-q)^k$, in the regime where $m$ is negligible compared to the number of rounds. 
\end{proof}

Let us define the half-life $t_{\Pi}(p)$ of a protocol $\Pi$ as the number of rounds required to achieve $P_\mathrm{ok}(\Pi) \approx 1/e$ if each station dies independently with probability $p$. 
Then, Proposition \ref{prob-ok} states that 
\begin{align}
t_{\FQ}(p) = \frac{1}{mp} \quad \text{and} \quad t_{\mathrm{Tree}}(p) = \frac{1}{q} \approx \frac{1}{3m^2p^2}
\end{align}
provided that $mp \ll 1$. 
In particular, adding a third player to the standard $\FQ$-protocol provides a quadratic improvement in the expected half-life of the commitment time.

Let us now evaluate the communication cost of the various protocols, that is the number of bits that are exchanged among various agents during the whole protocol. Note first that by construction, all the challenges and responses are elements of $\F_Q$, meaning that each round (corresponding to each alive node in the Tree protocol) has an individual cost of $2 \log_2 Q$ bits.

\begin{proposition}
The communication cost $C_{\F_Q}$ and $C_{\mathrm{Tree}}$ of the $k$-round $\F_Q$ and Tree protocols are given by:
\begin{align}
C_{\F_Q} &= 2k \log_2 Q\\
C_{\mathrm{Tree}}&\approx  k 2^{N+2} \log_2 Q,
\end{align}
where $N$ is the number of rounds necessary for all agents to realize that a given branch is dead. Recall that taking $\log_2 Q = O(\log (k/\eps))$ is sufficient to guarantee that the protocol is $\eps$-binding.
\end{proposition}

In practice, the value of $N$ will be a small constant, which shows that the communication cost of the Tree protocol compares favorably with that of the original $k$-round $\F_Q$ protocol.

\begin{proof}
Obtaining the communication cost of the $\F_Q$ protocol is straightforward: there are $k$ rounds that each cost $2\log_2 Q$ bits. 

For the Tree protocol, we consider the ``worst case scenario'' where Alice's agents always respond in time. This means that all branches are alive unless Bob's agents decide not to send them challenges anymore. Since only the leftmost alive branch matters in the reveal phase, and since the prefix of the leftmost alive node never changes during the protocol, it is easy to see that Bob's agents do not need to continue sending challenges to branches that they know not to be the leftmost alive branch. 
In general, it may take $N$ additional rounds before all agents learn the status of all the history up to a given round. This means that in the worst case, Bob's agents should send challenges to all the descendants of the current leftmost alive node for $N$ rounds. The number of such nodes is upper bounded by $2^{N+1}$. 
Since there are $k$ rounds in total, the communication cost of the Tree protocol $\mathrm{Tree}$ can be upper bounded by $2^{N+1}k \times 2 \log_2 Q$ bits. 
\end{proof}

\section{Conclusion}

In this paper, we introduced a new relativistic bit commitment protocol that addresses one of the main weaknesses of the $\FQ$ protocol, namely its fragility against network failures. Indeed, the $\FQ$ protocol aborts as soon as one agent fails to respond to a single challenge in time. 
We fix this issue by modifying the $\FQ$ protocol so that each party is now represented by 3 agents in 3 distinct locations. The communication cost of this variant is relatively modest, but the gain in terms of tolerance to loss is very good: one expects a quadratic gain for the number of rounds that the protocol can sustain, making it very promising for implementations in real telecom networks (instead of dedicated networks), which is crucial for a possible future deployment of this technology. 

We conclude with a couple of open problems that are left for future investigation. First, the tree structure that we rely on here does not seem to be optimal and simpler schemes with reduced communication complexity would be interesting. Second, our security analysis is restricted to classical adversaries, as was already the case in \cite{CCL15, FF15} and the obvious next step is to see whether one can also prove security against quantum adversaries. The main difficulty to extend the analysis to the quantum case is that the composition of the rounds is more complicated to handle because the history is not described by classical random variables anymore, but rather by quantum states.

\section*{Acknowledgements}

We are grateful to Fr\'ed\'eric Grosshans for stimulating discussions on relativistic cryptography.

\bibliographystyle{alpha}

\newcommand{\etalchar}[1]{$^{#1}$}

\newpage

\appendix

This appendix contains the proofs of the main technical claims as well as a short description of the generalization of the Tree protocol to an arbitrary number of agents per party.

\section{Proof of Sum inversions} \label{Apppendix:ProofOfSumInversions}

In this section, we prove Proposition \ref{lem:lemma1} which we recall below.

\begin{proposition} For $j \leq k-2$,
	$$
	IP_{j+1} = \E_{v \leftarrow Z_j} \E_{h_0 \in H_{j+1}^{-\{v,v\ell,vr\}}(v)}\E_{t \leftarrow T(v|h_0)} [IP_{vt}^{h_0}].
	$$
\end{proposition}

\begin{proof}
	Fix an integer $j$, a node $v \in V_{=j}$ and a history $h_1 \in H_{j+1}^{-\{v\ell,vr\}}(v)$. Let us define $T(v|h_1)$, the random variable equal to `$\ell$' with probability $\frac{|B_{j+1}^{h_1}(v\ell)|}{Q}$ and `$r$' with probability $1 - \frac{|B_{j+1}^{h_1}(v\ell)|}{Q}$. If $h_1$ is consistent with $v$, then $vt$ with $t =T(v|h_1)$ is the left-most alive node at depth $j+1$. Let us also define
	$$ C_t^{h_1}(v\ell) = \left\{\begin{array}{ll}
	B_{j+1}^{h_1}(v\ell) & \textrm{if } t = \ell \\
	\F_Q - B_{j+1}^{h_1}(v\ell) & \textrm{if } t = r
	\end{array}
	\right. $$
to be the set of possible values of $b_{v\ell}$ conditioned on the node $v\ell$ being responsive ($C_\ell$) or not ($C_r$).
	
By averaging over histories $h_1$ consistent with the node $v$, we define the random variable $T(v)$ equal to `$\ell$' with probability $\frac{|H_{j+1}(v\ell)|}{|H_{j+1}(v)|}$ and to `$r$' with probability $\frac{|H_{j+1}(vr)|}{|H_{j+1}(v)|} = 1 - \frac{|H_{j+1}(v\ell)|}{|H_{j+1}(v)|}$:
\begin{align}
T(v) := \E_{h_1 \leftarrow H_{j+1}^{- \{v\ell,vr\}}(v)}[T(v|h_1)].
\end{align}

	\begin{lemma}\label{Proposition:FirstExpectancies}
		\begin{align*}
			&IP_{j+1}(\{g_{v'}\}_{v' \in V_{=j+1}}) \\
			&=  \E_{v \leftarrow Z_{j}}  \E_{h_1 \leftarrow H_{j+1}^{- \{v\ell,vr\}}(v)} \E_{t \leftarrow T(v|h_1)}  \E_{b_{v\ell} \leftarrow C_{t}^{h_1}(v\ell)}
			\E_{b_{vr} \leftarrow \F_Q}[g_{vt}(d,h_1) == \alpha_{vt}(d,h_1,b_{vt})]
		\end{align*}

	\end{lemma}
	
	\begin{proof}
		
		According to the definition of $IP_{j+1}$ we have,	
		\begin{align*} 
 IP_{j+1}(\{g_{v'}\}_{v' \in V_{=j+1}}) 	& = \E_{v' \leftarrow Z_{j+1}} \E_{h \leftarrow H_{j+1}^{- \{v'\}}(v')} \E_{b_{v'} \leftarrow B_{j+1}^h(v')}[g_{v'}(d,h) == \alpha_{v'}(d,h,b_{v'})] \\ 
	& =  \E_{v \leftarrow Z_{j}} \E_{t \leftarrow T(v)} \E_{h \leftarrow H_{j+1}^{- \{vt\}}(vt)} \E_{b_{vt} \leftarrow B_{j+1}^h(vt)}[g_{vt}(d,h) == \alpha_{vt}(d,h,b_{vt})] 
	\end{align*}
The statement of the lemma follows from the fact that  $a_{vt}$ does not depend on $b_{v\overline{t}}$.	
	\end{proof}
	
	\begin{lemma}
		\begin{eqnarray*}
			IP_{j+1} & = & \E_{v \leftarrow Z_j} \E_{h_0 \in H_{j+1}^{-\{v,v\ell,vr\}}(v)}\E_{t \leftarrow T(v|h_0)} \\
			& & \max_{g_{vt}} \E_{b_v \in B_j^{h_0}(v)}  \E_{b_{vt} \in B_{j+1}^{(h_0,b_v)}(vt)} [g_{vt}(d,h_0,b_v) == \alpha_{vt}(d,h_0,b_v,b_{vt})].
		\end{eqnarray*}
	\end{lemma}
	
	\begin{proof}
		From Lemma \ref{Proposition:FirstExpectancies}, we have
		\begin{equation}
		\label{eqipjg}
		IP_{j+1}(\{g_{v'}\}) = \E_{v \leftarrow Z_{j}}  \E_{h_1 \leftarrow H_{j+1}^{- \{v\ell,vr\}}(v)} \E_{t \leftarrow T(v|h_1)}  \E_{b_{v\ell} \leftarrow C^{h_1}_{t}(v\ell)}\E_{b_{vr} \leftarrow F_Q}[g_{vt}(d,h_1,b_{v\overline{t}}) == \alpha_{vt}(d,h_1,b_{vt})] 
		\end{equation}
		From the definition of $IP_j$ we have,
		\begin{equation}
		IP_{j+1} = \max_{g_{vt} \in V_{=j+1}} IP_{j+1}(\{g_{v'}\})
		\end{equation}
		Since $a_{vt}(d,h_1,b_{vt})$ doesn't depend on $b_{v\overline{t}}$, the value of $IP_{j+1}$ remains unchanged if $g_{vt}$ depends only on $h_1$.
		This implies that we can write $IP_{j+1}$ as follows,
\begin{align*}
IP_{j+1}  &=  \max_{g_{vt} } \E_{v \leftarrow Z_{j}}  \E_{h_1 \leftarrow H_{j+1}^{- \{v\ell,vr\}}(v)} \E_{t \leftarrow T(v|h_1)}  \E_{b_{vt} \leftarrow B^{h_1}_{j+1}(vt)}[g_{vt}(d,h_1) == \alpha_{vt}(d,h_1,b_{vt})]\\
			 &=  \max_{g_{vt}} \E_{v \leftarrow Z_{j}}  \E_{(h_0,b_v) \leftarrow (H_{j+1}^{- \{v,v\ell,vr\}}(v)\times B^{h_0}_j)} \E_{t \leftarrow T(v|h_0,b_v)}  \E_{b_{vt} \leftarrow B^{h_0,b_v}_{j+1}(vt)}[g_{vt}(d,h_1) == \alpha_{vt}(d,h_1,b_{vt})]
\end{align*}
where $ h_1 = (h_0,b_v)$
\begin{align*}
IP_{j+1} &=  \max_{g_{vt}} \E_{v \leftarrow Z_{j}}  \E_{h_0 \leftarrow H_{j+1}^{- \{v,v\ell,vr\}}(v)} \E_{b_v \in B^{h_0}_{j}}\E_{t \leftarrow T(v|h_0,b_v)}  \E_{b_{vt} \leftarrow B^{h_0,b_v}_{j+1}(vt)}[g_{vt}(d,h_1) == \alpha_{vt}(d,h_1,b_{vt})]\\
			& =  \max_{g_{vt}} \E_{v \leftarrow Z_{j}}  \E_{h_0 \leftarrow H_{j+1}^{- \{v,v\ell,vr\}}(v)} \E_{b_v \in B^{h_0}_{j}}\E_{t \leftarrow T(v|h_0)}  \E_{b_{vt} \leftarrow B^{h_0,b_v}_{j+1}(vt)}[g_{vt}(d,h_0,b_v) == \alpha_{vt}(d,h_0,b_v,b_{vt})]
			\end{align*}
Notice that once we fix a leftmost alive node, the decision to go left or right is independent of $b_v$. Therefore, we have $T(v|h_0) = T(v|h_0,b_v)$, for any $b_v \in B^{h_0}_j(v)$.
\begin{align*}
IP_{j+1} &=   \E_{v \leftarrow Z_{j}}  \E_{h_0 \leftarrow H_{j+1}^{- \{v,v\ell,vr\}}(v)} \E_{t \leftarrow T(v|h_0)} \max_{g_{vt}} \E_{b_v \in B^{h_0}_{j}}\E_{b_{vt} \leftarrow B^{h_0,b_v}_{j+1}(vt)}[g_{vt}(d,h_0,b_v) == \alpha_{vt}(d,h_0,b_v,b_{vt})].
\end{align*}
	\end{proof}

	For a fixed history $h_0 \in H_{j+1}^{-\{v,v\ell,vr\}}(v)$ and $d$, we define the quantity $IP^{h_0}_{vt}$ in following manner,
	\begin{equation}
	IP_{vt}^{h_0,d}:= \max_{g^{h_0}}\E_{b_v \leftarrow B_j^{h_0}(v)} \E_{b_{vt} \leftarrow B_{j+1}^{h_0,b_v}(vt)} [g(d,h_0,b_v) == \alpha_{vt}(d,h_0,b_v,b_{vt})].
	\end{equation}
	Substituting the expression of $IP_{vt}^{h_0,d}$ in the expression of $IP_{j+1}$ we get,
	\begin{equation}
	\label{ipj1}
	IP_{j+1} = \E_{v \leftarrow Z_j} \E_{h_0 \in H_{j+1}^{-\{v,v\ell,vr\}}(v)}\E_{t \leftarrow T(v|h_0)} [IP_{vt}^{h_0,d}].
	\end{equation}
\end{proof}
\section{Base case of the recursion: $j=0$}\label{Appendix:j=0}
We first consider the base case where $j=0$.

\begin{lemma}
	$$IP_{0} \le \frac{1}{2} + \sqrt{\frac{2}{Q}}.$$
\end{lemma}

\begin{proof}
	According to the definition of $IP_j$ we have,
	\begin{equation}
	IP_j = \max_{\{g_v\}_{v \in V_{=j}}} IP_j(\{g_v\}_{v \in V_{=j}}),
	\end{equation}
	
	where, 
	
	\begin{equation}
	IP_j(\{g_v\}_{v \in V_{=j}}) = \E_{v \leftarrow Z_j} \E_{h \leftarrow H_j^{- \{v\}}(v)} \E_{b_v \leftarrow B_j^h(v)}[g_v(d,h) == \alpha_j(d,h,b_v)].
	\end{equation}

	For $j = 0$, \textit{i.e.}, at the root of the tree, we have $V_{=j} = \{ v_0\}$, where $v_0 = \varnothing$, $H_0^{- \{v_0\}}(v_0)$ contains only the commitment $d$ and $B_j^h(v) = \F_Q$. So, we have 
	$IP_{0} = \max_{g_{v_0}}\E_{d \leftarrow \{0,1\}} \E_{b_{v_0} \leftarrow \F_Q}[g_{v_0}(d) == \alpha_{v_0}(d,b_{v_0})]$. 
	Here we give the upper bound on $IP_{0}$ by reducing it to an instance $G$ of the following nonlocal games between two players Adeline and Bastian, where 
	\begin{itemize}
		\item Adeline receives a random element $b_{v_0} \in \F_Q$. 
		Bastian receives a random element $d \in \{0,1\}$. 
		\item Their goal is to respectively output $A$ and $B$ in $\F_Q$ such that $A + B  = b_{v_0} * d$.
	\end{itemize}
	
	Without any loss of generality we can consider Adeline and Bastian's strategy to be deterministic, namely Adeline's strategy is a deterministic function $y_{v_0}(b_{v_0})$ and Bastian's strategy is a deterministic function $-g_{v_0}(d)$. This strategy gives a lower bound on the value $\omega(G)$ of the game:
	\begin{align*}
	\omega(G) & \geq  \max_{g_{v_0}} \Pr_{b_{v_0},d} [y_{v_0}(b_{v_0}) - g_{v_0}(d) = b_{v_0} * d]\\
	& =  \max_{g_{v_0}} \Pr_{b_{v_0},d} [\alpha_{v_0}(d,b_{v_0}) + d * b_{v_0}) - g_{v_0}(d) = (b_{v_0} * d)] \\
	&  \mbox{(substituting $y_{v_0} = \alpha_{v_0} + b_{v_0}*d$)} \\
	& =  \max_{g_{v_0}} \Pr_{b_{v_0},d} [g_{v_0}(d) == \alpha_{v_0}(d,b_{v_0})]\\
	& =  IP_{0}.
	\end{align*}
	We can conclude using the result of Lemma \ref{lemma} proven in the next section to the case where $p=1/2$ and $S = \{0,1\}$: we obtain
	\begin{equation}
	IP_{0} \leq \frac{1}{2} + \sqrt{\frac{2}{Q}}.
	\end{equation}
\end{proof}

\section{A generalization of $\mathrm{CHSH}_Q(p)$ games with restricted inputs.}
\label{games}

The class of $\mathrm{CHSH}_Q(p)$ games was introduced in \cite{CCL15} in order to analyze the security of the $\FQ$ protocols. These are simply two-party nonlocal games between Adeline and Bastian who respectively receive inputs $x, y\in \F_Q$ and output $a, b \in F_Q$. Here $x$ is drawn from the uniform distribution while $y$ is drawn according to a probability distribution $\{p_y\}_{y\in \F_Q}$ such that $\max_y p_y \leq p$.
Adeline and Bastian win the game if $a + b = x *y$ in $\FQ$. 

Here, we define a slight variant of these games where the only difference is now that Adeline's inputs are drawn uniformly from a subset $S$ of $\F_Q$. We denote this class of games by  $\mathrm{CHSH}^S_Q(p)$. In particular, one has $\mathrm{CHSH}_Q(p) = \mathrm{CHSH}^{\F_Q}_Q(p)$.

It is straightforward to upper bound the classical value of games in $\mathrm{CHSH}^S_Q(p)$ using the same technique as in \cite{CCL15}. For completeness, we include this proof here.

\begin{lemma}
\label{lemma}
	For any game $G \in \mathrm{CHSH}_Q^S(p)$, we have 
	\begin{align}
	\omega(G) \le p + \sqrt{\frac{2}{|S|}}.
	\end{align}
\end{lemma}
\begin{proof}
Fix a game $G \in \mathrm{CHSH}_Q^{S}(p)$. As usual, the classical value of the game can always be achieved with a deterministic strategy, meaning that without loss of 
generality, Alice and Bob's strategies can be modeled by functions $f$ and $g$, namely: $a = f(x)$ and $b = g(y)$. 
Define the variable $r_x^{y}$ equal to $1$ if $f(x) + g(y) = x *y$ and $0$ otherwise.

Consider the following strategy for Bob: pick a random pair of distinct inputs $y, y'$ according to the distribution $\{p_y\}_{y \in \F_q}$, \textit{i.e.}~ 
with probability $p_y p_y'/P$ where $P = \sum_{y \ne y'} p_y p_y'$, and output the guess $\hat{x}$ for $x$ defined by $\hat{x} = (g(y)-g(y'))*(y-y')^{-1}$.
Let $S_x$ be the probability of correctly guessing the value $x$ with this strategy. Non signaling imposes that $\mathbbm{E}_x[S_x] = 1/|S|$, since the value $x$ is uniformly distributed in $S$.

On the other hand, we note that if the game $G$ is won for both inputs $(x,y)$ and $(x,y')$, then Bob's strategy outputs the correct value for $x$. Indeed, winning 
the game for both inputs means that $f(x) + g(y) = x*y$ and $f(x) + g(y') = x*y'$ which implies 
that $g(y)-g(y') = (y-y')*x$ and therefore $\hat{x} = x$.
One immediately obtains a lower bound on $S_x$:
\begin{align}
\label{eqn:bound}
S_x \geq \frac{1}{P} \sum_{y, y' \neq y} p_y r_x^y p_y' r_{x}^{y'} \geq  \sum_{y,y' \neq y} p_y r_x^y p_y' r_{x}^{y'}, 
\end{align}
where the second inequality follows from the fact that $P \leq 1$.
Consider the quantity $\omega^x = \sum_y p_y r_x^y$. It satisfies:
$$(\omega^x)^2 \leq \sum_y p_y^2 (r_x^y)^2 + 2 S_x = \sum_y (p_y)^2 r_x^y + 2 S_x \leq p \omega^x + 2S_x,$$
where the first inequality follows from the bound of Eq.~\ref{eqn:bound} and where we used that $(r_x^y)^2=r_x^y$ and $(p_y)^2 \leq \left(\max_{y} \{p_y\} \right) p_y \leq p p_y$. 
Solving this quadratic equation gives that
$$ \omega^x \leq \frac{1}{2}\left( p  +\sqrt{p^2 + 8S_x}\right)$$
and the concavity of the square-root function implies that
$$ \omega^x \leq  p + \sqrt{2S_x}.$$
Finally, $\omega(G) = \mathbbm{E}_x[\omega^x]$ by definition and using the concavity of the square-root function once more shows that:
$$\omega(G) \leq p + \sqrt{2}\mathbbm{E}_x [\sqrt{S_x}] \leq p + \sqrt{2} \sqrt{\mathbbm{E}_x[S_x]} \leq p +\sqrt{2/|S|},$$
which concludes the proof.
\end{proof}

\section{Generalization to $n$ agents per party}
\label{sec:n-agents}

It is straightforward to generalize the Tree protocol to the case where each party is represented by $n$ agents. In that case, the binary tree should be replaced by a complete $n$-ary tree, together with an $n$-coloring of that tree. 
For the protocol to abort, it requires that $n-1$ stations die simultaneously. It is straightforward to see that the probability that the protocol succeeds becomes $(1-q(n))^k$ with
\begin{align}
q(n) = n (mp)^{n-1} + (mp)^n. 
\end{align}
Provided that $nmp \ll1$, the half-life of the generalized Tree protocol $\mathrm{Tree}(n)$ with $n$ agents per player becomes:
\begin{align}
t_{\mathrm{Tree}(n)}(p,m)  \approx \frac{1}{n (mp)^{n-1}}.
\end{align}

It is less straightforward to generalize the security proof to the case of $n$ agents. However, it is natural to conjecture that an analysis similar to that of Proposition \ref{prop:rec} for the Tree protocol with 3 locations will work. 
\begin{conj}
The $k$-round Tree protocol with $n \geq 3$ agents per party is $\varepsilon_{k,n}$-binding with 
\begin{align}
\varepsilon_{k,n} = 2 k x_n \sqrt{\frac{2}{Q}}
\end{align}
with
\begin{align}
x_2 = 1, \quad x_n = x_{n-1} + \frac{1}{4 x_{n-1}}.
\end{align}
In particular, asymptotically, it holds that $x_n \sim \sqrt{n/2}$. 
\end{conj}

\end{document}